\documentclass[sigplan,screen]{acmart}

\usepackage{subcaption}
\usepackage{algorithm}
\usepackage[noend]{algpseudocode}

\usepackage{enumitem}
\usepackage{multirow}

\theoremstyle{definition}
\newtheorem{definition}{Definition}[section]
\newtheorem{theorem}{Theorem}[section]
\newtheorem{lemma}[theorem]{Lemma}

\usepackage[flushleft]{threeparttable}

\usepackage{listings}

\definecolor{dkgreen}{rgb}{0,0.6,0}
\definecolor{gray}{rgb}{0.5,0.5,0.5}
\definecolor{mauve}{rgb}{0.58,0,0.82}

\usepackage{microtype}

\AtBeginDocument{%
  }

\setcopyright{cc}
\setcctype{by}
\acmDOI{10.1145/3708493.3712695}
\acmYear{2025}
\copyrightyear{2025}
\acmISBN{979-8-4007-1407-8/25/03}
\acmConference[CC '25]{Proceedings of the 34th ACM SIGPLAN International Conference on Compiler Construction}{March 1--2, 2025}{Las Vegas, NV, USA}
\acmBooktitle{Proceedings of the 34th ACM SIGPLAN International Conference on Compiler Construction (CC '25), March 1--2, 2025, Las Vegas, NV, USA}
\received{2024-11-12}
\received[accepted]{2024-12-21}

\begin{document}

\settopmatter{printfolios=false}

\title{Compiler Support for Speculation in Decoupled Access/Execute Architectures}

\author{Robert Szafarczyk}
\orcid{0009-0007-8883-1747}
\affiliation{%
  \institution{University of Glasgow}
  \city{Glasgow}
  \country{United Kingdom}
}
\email{robert.szafarczyk@glasgow.ac.uk}

\author{Syed Waqar Nabi}
\orcid{0000-0003-3835-4851}
\affiliation{%
  \institution{University of Glasgow}
  \city{Glasgow}
  \country{United Kingdom}
}
\email{syed.nabi@glasgow.ac.uk}

\author{Wim Vanderbauwhede}
\orcid{0000-0001-6768-0037}
\affiliation{%
  \institution{University of Glasgow}
  \city{Glasgow}
  \country{United Kingdom}
}
\email{wim.vanderbauwhede@glasgow.ac.uk}






\begin{abstract}
Irregular codes are bottlenecked by memory and communication latency.
Decoupled access/execute (DAE) is a common technique to tackle this problem.
It relies on the compiler to separate memory address generation from the rest of the program, however, such a separation is not always possible due to control and data dependencies between the access and execute slices, resulting in a loss of decoupling. 

In this paper, we present compiler support for speculation in DAE architectures that preserves decoupling in the face of control dependencies.
We speculate memory requests in the access slice and poison mis-speculations in the execute slice without the need for replays or synchronization.
Our transformation works on arbitrary, reducible control flow and is proven to preserve sequential consistency.
We show that our approach applies to a wide range of architectural work on CPU/GPU prefetchers, CGRAs, and accelerators, enabling DAE on a wider range of codes than before.
\end{abstract}

\maketitle

\section{Introduction}
\begin{figure}[t]
\centering
 \begin{subfigure}[b]{0.45\textwidth}
     \centering
     \includegraphics[width=\textwidth]{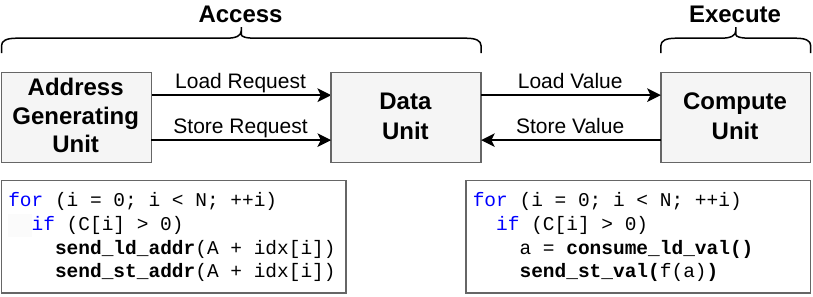}
     \caption{An architecture with decoupled address generation, memory access, and compute.}
     \label{fig:DecoupledAccessExecute_a}
 \end{subfigure}
 \hfill
 \begin{subfigure}[b]{0.45\textwidth}
     \centering
     \includegraphics[width=\textwidth]{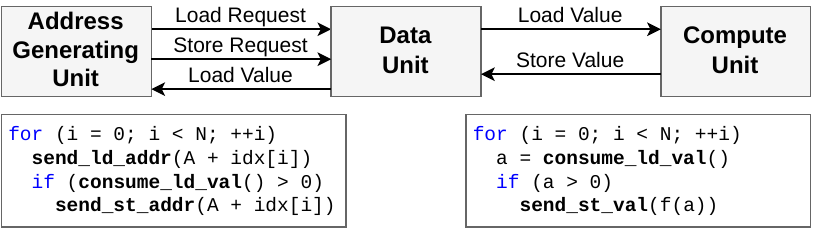}
     \caption{Loss-of-decoupling between address generation and memory access due to a dependency on the memory value.}
     \label{fig:DecoupledAccessExecute_b}
 \end{subfigure}
 \hfill
 \begin{subfigure}[b]{0.45\textwidth}
     \centering
     \includegraphics[width=\textwidth]{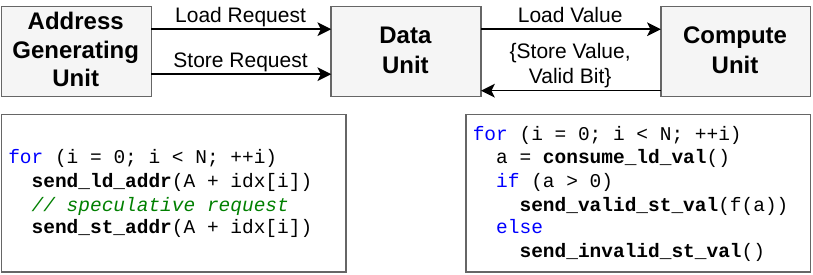}
     \caption{Our contribution: compiler support for speculation removes loss-of-decoupling due to control dependencies.}
     \label{fig:DecoupledAccessExecute_c}
 \end{subfigure}

\caption{A decoupled access/execute architecture template.}
\label{fig:DecoupledAccessExecute}
\end{figure}

Irregular codes are characterized by data-dependent memory accesses and control flow, for example:
\begin{lstlisting}
    for (int i = 0; i < N; ++i)
        if (C[i] > 0) 
            A[idx[i]] = f(A[idx[i]]);
\end{lstlisting}
This code has unpredictable control flow that causes frequent branch mis-predictions on CPUs and thread divergence on GPUs.
Because of these limitations, and challenges with Moore's Law and Dennard performance scaling, computer architects are interested in adding CPU/GPU structures to accelerate such code patterns, or even to use accelerators specialized for a given algorithm \cite{golden_age_for_arch}.

Many of the proposed architectures follow the decades-old idea of a \textit{Decoupled Access/Execute} (DAE) architecture shown in Figure~\ref{fig:DecoupledAccessExecute}.
In DAE, memory accesses are \textit{decoupled} from computation to avoid stalls resulting from unpredictable loads \cite{decoupled_access_exec}.
The address generation unit (AGU) sends load and store requests to the data unit (DU), while the DU sends load values to and receives store values from the compute unit (CU).
All communication is FIFO based and ideally the AGU to DU communication is one-directional, allowing the address streams from the AGU to run ahead w.r.t the CU.
Figure~\ref{fig:DecoupledAccessExecute_a} shows an example of such a DAE architecture implementing the earlier code snippet.

DAE is a general technique applicable to many computational models:
it is used in specialized FPGA accelerators generated from High-Level Synthesis (HLS) \cite{decoupled_memory_prefetching, dae_accel, decoupled_memory_wawrzynek, brainwave_ms_chung2018serving, ThunderGP, HLS_runahead, szafarczyk_fpt_23, szafa_fpga25}; 
in Coarse Grain Reconfigurable Architectures (CGRAs) \cite{fifer, Plasticine_rdu, SambaNova, fan2023_europar_decoupled_dataflow, hong2020decoupling_cgra, pellauer2019buffets, dae_cgra_cascade, cgra_dae_softbrain};
and in CPU/GPU prefetchers \cite{outrider_decoupled_strands, phloem, desc_cpu, gpu_decoupled, nvidia_hopper_tma, wasp_gpu_prefetch, qin2023roma, dae_cpu_stream_nowatzki_isca19}.
For example, NVIDIA introduced hardware-accelerated asynchronous memory copies \cite{nvidia_hopper_tma}.
The CUDA programmer can provide a ``copy descriptor'' of a tensor to copy and the hardware will run ahead and generate the corresponding addresses in a Tensor Memory Unit.




The common denominator of all these works is that they rely on either the programmer or the compiler to decouple address-generating instructions from the rest of the program.
However, it has long been recognized that such a decoupling is not always possible \cite{effectiveness_of_decoupling_sc, compiling_and_opt_for_dae}.
If any of the address-generating instructions for array \texttt{A} depend on a value loaded from \texttt{A}, then there is a \textit{loss-of-decoupling} (LoD) \cite{loss_of_decoupling}.
Access patterns such as \texttt{A[f(A[i])]} are rare, but control dependencies that involve loads from \texttt{A} are common.
For example, consider replacing \texttt{C[i]} with \texttt{A[i]} in our running example:
\begin{lstlisting}
    for (int i = 0; i < N; ++i)
        if (A[i] > 0) 
            A[idx[i]] = f(A[idx[i]]);
\end{lstlisting}
Here, there is a LoD, because the \texttt{A} store is control-dependent on a branch that loads from \texttt{A}.
Whereas before the load from \texttt{C} could be prefetched, now the AGU/DU communication is synchronized, because the AGU waits for \texttt{A} values from the DU before deciding if a store address should be generated, as shown in Figure~\ref{fig:DecoupledAccessExecute_b}.
In turn, the load waits for the store address to ensure that there is no aliasing---the store address is needed for memory disambiguation.
As a result, the AGU cannot run ahead of the CU anymore, resulting in decreased pipeline parallelism, which Figure~\ref{fig:Pipeline} illustrates.

\begin{figure}[t]
\centering
 \begin{subfigure}[b]{0.475\textwidth}
     \centering
     \includegraphics[width=\textwidth]{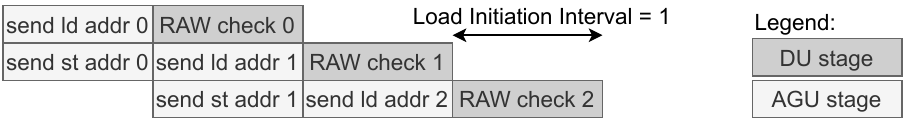}
     \caption{Pipeline of decoupled address generation from Figure \ref{fig:DecoupledAccessExecute_a}.}
     \label{fig:PipelineA}
 \end{subfigure}
 \hfill
 \begin{subfigure}[b]{0.475\textwidth}
     \centering
     \includegraphics[width=\textwidth]{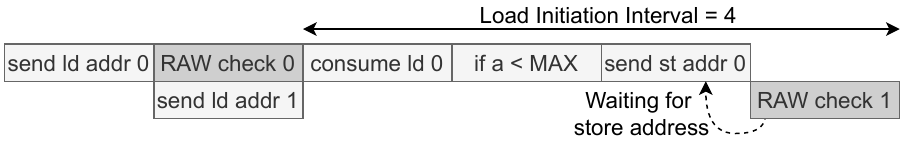}
     \caption{Pipeline of non-decoupled address generation from Figure \ref{fig:DecoupledAccessExecute_b}.}
     \label{fig:PipelineB}
 \end{subfigure}

\caption{Comparison of a decoupled and non-decoupled address generation. Non-decoupled address generation results in a later arrival of the store address, which stalls the RAW check for the next load, lowering load throughput.}
\label{fig:Pipeline}
\end{figure}




One approach for restoring decoupling in this case is control speculation.
As shown in Figure~\ref{fig:DecoupledAccessExecute_c}, we can hoist the store request out of the \textit{if}-condition in the AGU (speculation), and later \textit{poison} the store in the CU on mis-speculation (store invalidation).
However, it is unclear how the compiler should coordinate the speculation and recovery transformations across two distinct control-flow graphs.
While the example from Figure~\ref{fig:DecoupledAccessExecute_c} is trivial, the task quickly becomes complicated with more speculated stores and nested control flow, as we demonstrate in the next section.
The \textit{key challenge} here is to guarantee that the order of store requests sent from the AGU matches the order of store values or kill signals sent from the CU on all control-flow paths.

General compiler support for speculated stores in DAE architectures is an open question that we tackle in this paper, making the following contributions:
\begin{itemize}
    \item We give a formal description of the fundamental reasons why address generation cannot always be decoupled from the rest of the program (\S\ref{sec:LoDAnalysis}).
    \item We describe compiler support for speculative memory in DAE architectures, solving the LoD problem due to control dependencies. We propose two algorithms: one for speculating memory requests in the AGU, and one for poisoning mis-speculations in the CU (\S\ref{sec:compiler_support}).
    \item We prove that our speculation approach preserves the sequential consistency of the original program and does not introduce deadlocks (\S\ref{sec:proof}).
    \item We show that our work enables DAE on a wider class of codes than before, with applications in CPU/GPU prefetchers, CGRAs, and FPGA accelerators.
    \item We evaluate our DAE speculation approach on accelerators generated from HLS implementing codes from the graph and data analytics domain. We achieve an average $1.9\times$ (up to $3\times$) speedup over the baseline HLS implementations. We show that our approach has no mis-speculation penalty and minimal code size impact (average accelerator area increase $<5\%$) (\S\ref{sec:evaluation}).
\end{itemize}

\section{Motivating Example} \label{sec:ordering_problem}

In this brief section, we show why an obvious approach to speculation in DAE architectures is incorrect.

The FIFO-based nature of DAE requires that the order of memory requests (speculative or not) generated in the AGU matches exactly the order of load/store values (poisoned or not) in the CU.
The motivating example in Figure~\ref{fig:DecoupledAccessExecute_c} contains just one speculative store and one path through the compute CFG where the speculation becomes unreachable, making the problem of ordering trivial in that case.

Consider the more complex code in Figure~\ref{fig:ComplexDecouplingExample_a} with three stores $s_0$, $s_1$, and $s_2$.
Speculating all store requests in the AGU might result in the store request order $(s_2, s_0, s_1)$.
In the CU, we need to guarantee the same order of corresponding store values (poisoned or not) on every possible control-flow path through the loop.
Unfortunately, the obvious approach that worked for the trivial example in Figure~\ref{fig:DecoupledAccessExecute_c} does not work here. 
If we poison values at points where the corresponding speculation becomes unreachable, as illustrated in Figure~\ref{fig:ComplexDecouplingExample_b}, then we end up with three possible orderings of store values depending on the CFG path in the CU, but only one of the orderings is correct.
This is why any previous implementations of speculative stores in DAE architectures has only considered trivial triangle or diamond shaped CFGs \cite{desc_cpu}, like the one in Figure~\ref{fig:DecoupledAccessExecute_c}.
Generalized compiler support for store speculation that guarantees the correct order of poisoning is the \textit{key challenge} that we solve in this paper.

\begin{figure}[t]
\centering
 \begin{subfigure}[b]{0.47\textwidth}
     \centering
     \includegraphics[width=\textwidth]{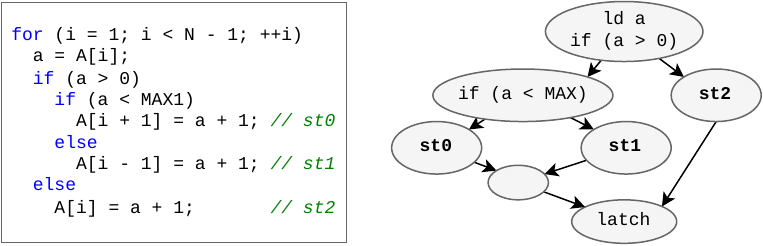}
     \caption{Code and control-flow graph of a loop with three control-dependent stores causing a loss-of-decoupling.}
     \label{fig:ComplexDecouplingExample_a}
 \end{subfigure}
 \hfill
 \begin{subfigure}[b]{0.47\textwidth}
     \centering
     \includegraphics[width=\textwidth]{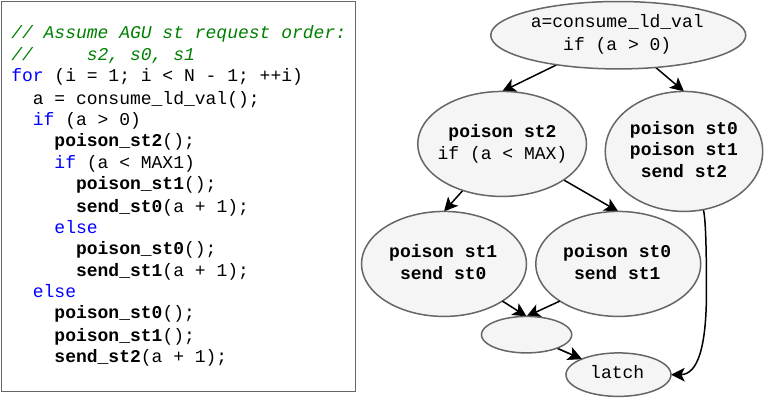}
     \caption{Depending on control flow, the order of store values can be: $(s_2, s_1, s_0)$, $(s_2, s_0, s_1)$, $(s_0, s_1, s_2)$, but only $(s_2, s_0, s_1)$ is correct.}
     \label{fig:ComplexDecouplingExample_b}
 \end{subfigure}

\caption{Poisoning speculated stores immediately when they become unreachable results in an ordering mismatch between AGU store requests from and CU store values.}
\label{fig:ComplexDecouplingExample}
\end{figure}


\section{Background} \label{sec:background}

In this section, we describe the architectural support needed to enable speculative DAE and some compiler preliminaries. 

\subsection{Architectural Support} \label{sec:accelerator_model}

Our speculation technique requires architectural support for predicated stores and FIFOs.
Store values are tagged with a \textit{poison bit} that, when set, causes the corresponding store request to be dropped in the DU without committing a store.
We say that a store request gets killed (or poisoned) if its corresponding store value has the poison bit set.
This is a lightweight form of speculation that does not require replays in the CU and does not result in out-of-bounds stores, because \textit{mis-speculated stores are never committed}.
Speculative loads can be supported by simply discarding the value of a mis-speculated load.

Predicated stores are easy to support in hardware since the underlying memory protocol usually already uses a \textit{valid} signal.
For example, the commonly used AXI4 interface \cite{AXI_ARM} has a strobe signal to indicate which bytes are valid.
Architectural FIFOs (queues) are also commonly added in works on CPU/GPU microarchitecture or can be relatively cheaply implemented in software. 
For example, works on DAE CPU/GPU prefetchers add architectural FIFOs and extend the ISA with instructions for producing load/store addresses and consuming/producing store values \cite{outrider_decoupled_strands, phloem, desc_cpu, gpu_decoupled, nvidia_hopper_tma, wasp_gpu_prefetch, qin2023roma}.

The prefetcher from \cite{desc_cpu} enables predicated stores with a \textit{store\_inv} instruction, but the authors support only simple triangle or diamond control flow patterns, calling for future work on general speculation support.
We discuss concrete examples of architectures that can benefit from our work in \S\ref{sec:applications}.
We evaluate our work on accelerators generated from HLS, where we have complete control of the memory interfaces.




\subsection{Compiler Preliminaries} \label{sec:compiler_Preliminaries}

We use an SSA-based compiler representation and associated analyses \cite{ssa_book}.
In particular, we use the control-flow graph and dominator tree to calculate control dependencies \cite{Ottenstein1990ThePD}, and we use the SSA def-use chain for data dependencies.

We use a canonical loop representation: loops have a single header block and a single loop backedge going from the loop latch to the loop header.
Our transformation assumes reducible control flow---CFG edges can be partitioned into two disjoint sets, forward and back edges, such that the forward edges form a directed acyclic graph (DAG).
Irreducible control flow can be made reducible with node splitting \cite{unstructured_to_structured_1973, unstructured_to_structured}.

For completeness, we briefly describe how our compiler implements a DAE architecture:
\begin{enumerate}
    \item \textbf{AGU:} For each memory operation to be decoupled, we change it to a \texttt{send\_ld\_addr} or a \texttt{send\_st\_addr} function that sends the memory address to the \textbf{DU}. 
    \item \textbf{CU:} Dually, in the CU we change each decoupled memory operation to a \texttt{consume\_val} or \texttt{produce\_val} function that receives or sends values to or from the DU. 
    \item \textbf{Dead Code Elimination:} We run a standard DCE pass in the CU to remove the now unnecessary address generation code. In the AGU, we delete all side effect instructions that are not part of the address generation def-use chains, and then also run a standard DCE pass. We also use a control-flow simplification pass that removes empty blocks potentially created by DCE.
\end{enumerate}

The \texttt{send\_ld\_addr}, \texttt{send\_st\_addr}, \texttt{consume\_val}, and\linebreak \texttt{produce\_val} functions are implementation dependent.
For example, if we target CPU/GPU prefetchers, such as \cite{desc_cpu, wasp_gpu_prefetch}, then these would translate to instructions.
For accelerators, they would be translated to FIFO writes/reads.




\section{Loss-of-Decoupling Analysis} \label{sec:LoDAnalysis}

LoD events arise when the address generation for a given memory access depends on a load that cannot be trivially prefetched, causing the AGU, DU, and CU communication to be synchronized.
By \textit{non-trivially prefetched} we mean loads that have a RAW hazards, i.e., the DU needs to receive all previous store addresses in program order to perform memory disambiguation before executing the load.

Given a set of address-generating instructions $G$, and a set of memory load instructions $A$ using addresses generated by instructions in $G$, there is a loss of decoupling if:

\begin{definition}[\textbf{LoD Data Dependency}]
There exists a path in the def-use chain from $a \in A$ to $g \in G$. While encountering a $\phi$-node on the def-use chain leading to $g$, we also trace the def-use paths of the terminator instructions $T$ in the $\phi$-node incoming basic blocks to see if any terminator instruction in $T$ depends on any $a \in A$.
\end{definition}

\begin{definition}[\textbf{LoD Control Dependency}]
There exists an instruction $g \in G$ that is control-dependent on a branch instruction $b$, and there is a path in the def-use chain from $a \in A$ to $b$. We call the basic block that contains $b$ the \textit{LoD control dependency source}. Note that the LoD control dependency source need not be the immediate control dependency of $g$, and that $g$ might have multiple LoD control dependencies.
\end{definition}

Depending on the hardware context, the definition of the $A$ set can be expanded or narrowed.
For example, if the AGU is implemented in hardware with limited control flow support, then $A$ could include all branch instructions.
On the other hand, given an address generating instruction, we could limit $A$ to only include loads from the same array for which the given address is generated---this could be useful if we only want to preserve decoupling for that array and do not care about losing decoupling for other arrays.
Our speculation technique applies equally well to all these definitions.

An example of a LoD data dependency is the \texttt{A[f(A[i])]} access.
Our speculation approach does not recover decoupling for such cases, but fortunately such accesses are rare.
An example of a more common LoD data dependency is the code pattern \texttt{if (A[i]) A[i++] = 1}.
In this case, the def-use chain leading to the definition of the store address contains a $\phi$-node (\texttt{i}) whose value depends on loading from $A$.
Such a pattern is sometimes found in algorithms that operate on dynamically growing data structures, e.g. queues or stacks.
Our speculation technique does not work on such cases either, but this is not a large limitation, since performance oriented codes typically do not use dynamically growing structures, instead opting for implementations with bounded space requirements that can be allocated statically \cite{syn_adt}.

An example of LoD due to a control dependency is shown in Figure~\ref{fig:DecoupledAccessExecute_b}.
This case is much more common than a direct data dependency and is the focus of this paper.


\section{Compiler Support for Speculation} \label{sec:compiler_support}

We now describe our dual transformations that enable speculation in the AGU and poison mis-speculations in the CU.

\begin{figure*}[ht]
\centering
\includegraphics[width=0.99\textwidth]{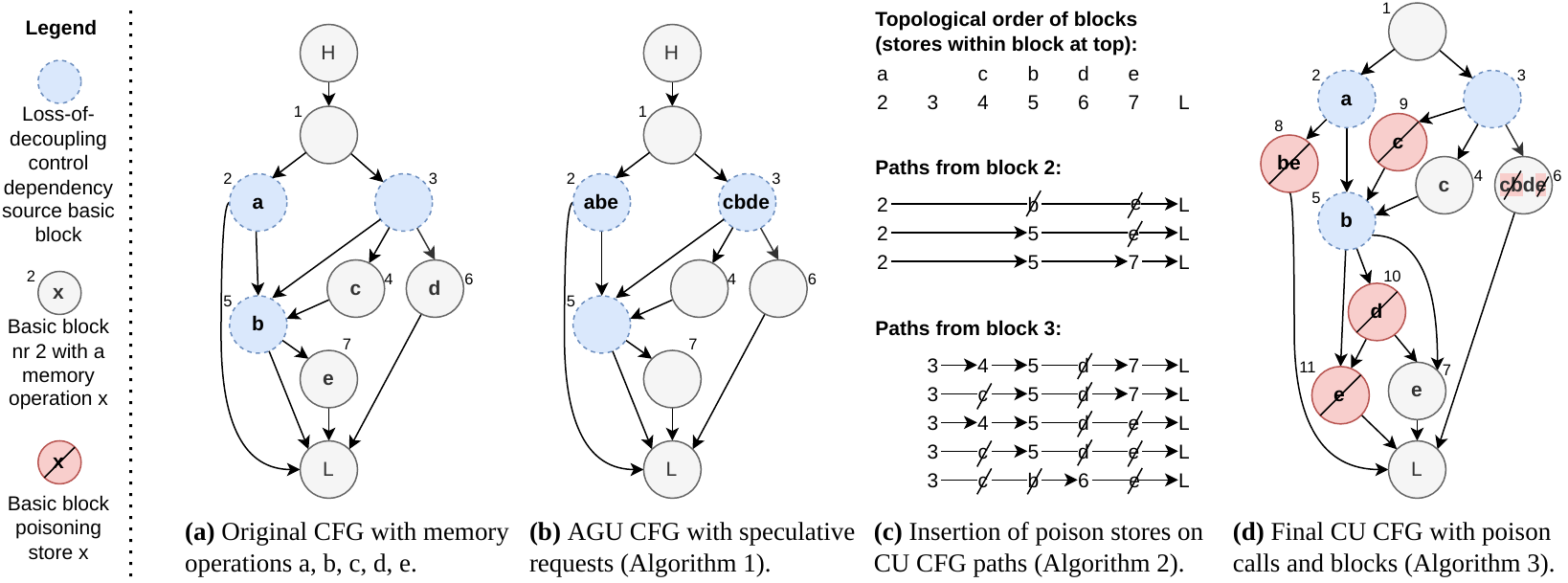}

\caption{An example of introducing speculative memory requests in the AGU (\S \ref{sec:Speculative_Memory_Requests}); and poisoned stores in the CU (\S \ref{sec:Poisoning_mis_speculated_Stores}). Block 6 in subfigure (d) kills stores c, b, then uses the allocation for store d, and then kills store e.}
\label{fig:HoistingAndPoisoning}
\end{figure*}



\algrenewcommand\algorithmicindent{1.5em}%

\begin{algorithm}[t]
\caption{Control-flow hoisting of AGU requests}\label{alg:speculation}
\begin{algorithmic}[1]

\State \textbf{Input:} $srcBlocks$ list of blocks that are the source of a LoD control dependency (defined in \S \ref{sec:LoDAnalysis}) 

\State \textbf{Output:} $SpecReqMap$ \{ basic block: list of hoisted requests to this block \} \\


\For{$srcBB \in srcBlocks$}
\\ \Comment{\textit{traverse the DAG from $srcBB$ to the loop latch}}
    \For{$fromBB \in reversePostOrder(srcBB)$}
        \If{$fromBB$ contains memory requests}
            \State hoist $fromBB$ requests to the end of $srcBB$
            \State add requests to $SpecReqMap[srcBB]$
        \EndIf
    \EndFor

\EndFor
    



\end{algorithmic}
\end{algorithm}

\subsection{Speculating Memory Requests} \label{sec:Speculative_Memory_Requests}

Algorithm~\ref{alg:speculation} describes our approach to introducing speculation in the AGU.
Given a LoD control dependency source block $srcBB$ we hoist all memory requests that are control dependent on $srcBB$ to the end of $srcBB$.
There can be multiple blocks with memory requests that have a LoD control dependency on $srcBB$, which poses the question in which order should they be hoisted to $srcBB$.
We use \textit{reverse post-order} in Algorithm~\ref{alg:speculation}.

Assuming reducible control flow, the CFG region from $srcBB$ to the loop latch is a DAG.
The reverse post-order of a DAG is its topological order.
Topological ordering gives us the useful property that given two distinct basic blocks $A$ and $B$ in a given loop, if $A \prec B$ in any path through the loop then $A \prec B$ in the topological ordering.
Note that there can be multiple topological orderings for a DAG, but it does not matter which one is chosen in our algorithm.

Algorithm~\ref{alg:speculation} traverses the CFG region from $srcBB$ to the end of its loop (or to the end of the function if $srcBB$ is not in a loop).
During the traversal, we ignore CFG edges leading to loop headers---we do not enter loops other than the innermost loop containing $srcBB$.

\subsubsection{Example of Hoisting}
Consider the CFG from Figure~\ref{fig:HoistingAndPoisoning}a.
There are three LoD control dependency source blocks ($2, 3, 5$) and five blocks with memory requests (blocks $2, 4, 5, 6, 7$ with requests $a, c, b, d, e$, respectively).
Assume that each block holds a single memory request---multiple memory requests within the same block are treated in the same way by our algorithms.
Figure~\ref{fig:HoistingAndPoisoning}c shows the topological order of the loop (block $1$ is omitted for brevity).
Algorithm~\ref{alg:speculation} will hoist $b, e$ to the end of block $2$, and $c, d, e$ to the end of block $3$---the result is presented in Figure~\ref{fig:HoistingAndPoisoning}b.
Note that the requests $b$ and $e$ were hoisted to both block $2$ and $3$, because they are reachable from both blocks.
Nothing is hoisted to block $1$ since it is not a LoD control dependency source.

\subsubsection{Nested LoD Control Dependencies} 
Block $5$ in Figure~\ref{fig:HoistingAndPoisoning}b does not contain any speculative requests because it itself has a LoD control dependency on block $2$ and $3$.
Algorithm~\ref{alg:speculation} considers only LoD control dependency source blocks that are not themselves the destination of another LoD control dependency.
Given a chain of nested LoD control dependencies, we only consider the chain head.
For example, the Figure~\ref{fig:HoistingAndPoisoning}a CFG has two LoD control dependency chains: $2, 5$ and $3, 5$---Algorithm~\ref{alg:speculation} considers only blocks $2$ and $3$.

\subsubsection{Why Topological Order in Algorithm~\ref{alg:speculation}?\nopunct} 
Topological order is needed to make it possible to match the order of speculative requests made in the AGU with the order of values that will arrive from the CU on all its possible CFG paths.
Consider, for example, the requests $b$ and $c$ in Figure~\ref{fig:HoistingAndPoisoning}a.
We first want to hoist $c$ to block $3$ before hoisting $b$, because there exists a CFG path where $c$ comes before $b$, but not vice versa.
If $b$ were hoisted before $c$, then the speculative requests order would be $b \prec c$, which would be impossible to match with values in the CU on the CFG path $3, 5, 7$.

\algrenewcommand\algorithmicindent{1.2em}%

\begin{algorithm}[t]
\caption{Mapping Poison Stores to CFG Edges in CU}\label{alg:poison}
\begin{algorithmic}[1]

\State \textbf{Input:} $SpecReqMap$ \{ basic block: list of requests hoisted to this block in Algorithm \ref{alg:speculation} \}
\\

\For{$specBB,\, specRequests \in SpecReqMap$}
    \For{$path \in allPathsToLoopLatch(specBB)$}
        \State $trueBlocks \gets \varnothing$ 
        \Comment{\textit{set keeps insertion order}}
        \For{$r \in specRequests$}
            \State $trueBB$ $\gets$ block where $r$ is true
            \State $trueBlocks.insert(trueBB)$
        \EndFor

        \For{$edge \in path$}
            \For{$trueBB \in trueBlocks$}
            
                \If{$edge_{dst} = trueBB$}
                    \State $trueBlocks.remove(trueBB)$ 
                    \State \textbf{break}
                    \Comment{\textit{to the next edge}}
                \EndIf
                
                \If{$trueBB$ not reachable from $edge_{dst}$}
                    \\ \Comment{\textit{reachability ignores loop backedges}}
                    \State poison $trueBB$ requests on $edge$
                    \Comment{\textit{Alg. \ref{alg:poison_edge}}}
                    
                    
                    \State $trueBlocks.remove(trueBB)$ 
                \EndIf
                
            \EndFor
        \EndFor
    \EndFor
\EndFor


\end{algorithmic}
\end{algorithm}

\subsection{Poisoning Mis-speculated Stores} \label{sec:Poisoning_mis_speculated_Stores}

Our strategy for poisoning misspeculations in the CU is to first map a poison call to a CFG edge, and then to map that edge to a poison store call contained in an existing or newly created basic block.

Algorithm~\ref{alg:poison} describes the first step.
Given block $specBB$ that contains speculative memory requests $specRequests$, we consider each path in the DAG from the $specBB$ to the loop latch (or function exit) in the CU.
We call the block where a $r \in specRequests$ becomes true the $trueBB$ (for example, the $trueBB$ for request $b$ in Figure~\ref{fig:HoistingAndPoisoning}a is block $5$).
For each CFG path, we use the $trueBlocks$ list to keep track of which requests were already used or poisoned on the path---the list contains the $trueBB$ for each $r \in specRequests$.

Given an edge in the traversal, the edge is skipped if the next $trueBB \in trueBlocks$ is still reachable from $edge_{dst}$.
This guarantees that the order of speculative requests in the AGU matches the order of values in the CU, i.e., a speculative request for a given $trueBB$ block is not poisoned immediately when $trueBB$ becomes unreachable if there is an earlier speculative request that can still be used.

\subsubsection{Example of Mapping Poison Stores to CFG Edges}
Figure~\ref{fig:HoistingAndPoisoning}c shows which CFG edges are poisoned given the original CFG in Figure~\ref{fig:HoistingAndPoisoning}a and the AGU CFG in Figure~\ref{fig:HoistingAndPoisoning}b.
For example, the path $3 \rightarrow 5 \rightarrow L$ will have: $poison(c)$ on the $3 \rightarrow 5$ edge; and $poison(d)$, $poison(e)$ on the $5 \rightarrow L$ edge (4th path from block 3 in Figure~\ref{fig:HoistingAndPoisoning}c).

\subsubsection{Mapping Poisoned Edges to Basic Blocks} \label{sec:poison_store_insertion}

\algrenewcommand\algorithmicindent{1.5em}%

\begin{algorithm}[t]
\caption{Poisoning Stores on Edges in CU}\label{alg:poison_edge}
\begin{algorithmic}[1]

\State \textbf{Input:} store request $r$; CFG $edge$; block $specBB$ where $r$ was speculated; block $trueBB$ where $r$ is true \\

\State $poisonBlockReuse \gets \varnothing$
\Comment{\textit{preserve set across calls}}
\If{$edge_{dst}$ is reachable from $trueBB$}
    \State $poisonBB \gets$ create new block on $edge$ \textbf{or} 
    \State \hspace{53pt} get from $poisonBlockReuse$ if exists
    \State append $poison(r)$ to the end of $poisonBB$
    \State $poisonBlockReuse.insert(poisonBB)$ 
\ElsIf{$specBB$ does not dominate $edge_{dst}$}
    \State $poisonBB \gets$ create new block on $edge$
    \State append $poison(r)$ to the end of $poisonBB$
    \\ \Comment{\textit{create recursively on $specBB \rightarrow edge_{src}$ paths}}
    \State create $\phi(1,\, specBB)$ value in $edge_{src}$
    \State branch from $edge_{src}$ to $poisonBB$ on $\phi = 1$
\Else
    \State append $poison(r)$ to the start of $edge_{dst}$
\EndIf
                    
                    
                

\end{algorithmic}
\end{algorithm}

Algorithm~\ref{alg:poison_edge} shows how poisoned CFG edges are mapped to actual poison calls placed in a concrete basic block.
Given a poisoned request $r$ on $edge$ (from $edge_{src}$ block to $edge_{dst}$ block), the $specBB$ block where $r$ was speculated in the AGU, and $trueBB$ where r becomes true there are three cases:
\begin{enumerate}
    \item There exists a path from $trueBB$ to $edge_{dst}$. In this case, we cannot insert $poison(r)$ in $edge_{dst}$, because we would end up with a CFG path where the store is both true and poisoned. To avoid this, we create a new $poisonBB$ block on $edge$ and append $poison(r)$ to it.
    \item There exists a path from the loop header to $edge_{dst}$ that does not contain $specBB$. In this case, we cannot insert $poison(r)$ in $edge_{dst}$, because we would end up with a CFG path where $r$ was not speculated in the AGU, but was poisoned in the CU. To avoid this, we create a new block $poisonBB$ on the edge and append $poison(r)$ to it. We also add steering instructions to the path from $specBB$ to $poisonBB$ that will branch from $edge_{src}$ to $poisonBB$ only if $specBB$ was encountered on the current CFG path.
    \item Otherwise, $poison(r)$ can safely be prepended to the start of $edge_{dst}$.
\end{enumerate}
Algorithm~\ref{alg:poison_edge} is executed only once per $(edge$,\, $r)$ tuple---a given request is poisoned at most once on a given edge.
Also, poison blocks created in case 1 in Algorithm~\ref{alg:poison_edge} can be reused to poison other requests.

\subsubsection{Example of Mapping Poison Edges to Blocks}

Consider how the poisoned edges in Figure~\ref{fig:HoistingAndPoisoning}c are mapped to basic blocks in Figure~\ref{fig:HoistingAndPoisoning}d.

\textit{Case 1:} Store $c$ is poisoned on the $3 \rightarrow 5$ edge.
Since there is a path from the true block of $c$ (block $4$) to the edge destination block (block $5$), we create a new block on the $3 \rightarrow 5$ edge and append $poison(c)$ to it.

\textit{Case 2:} Store $d$ is poisoned on both the $5 \rightarrow 7$ and $5 \rightarrow L$ edges.
The $specBB$ for $d$ is block $3$.
Since there exists the path $H \rightarrow 1 \rightarrow 2 \rightarrow 5$ that does not contain block $3$, we create a new block on the $5 \rightarrow 7$ edge with the $poison(d)$ call. 
We add steering instructions to the $3 \rightarrow 5$ and $3 \rightarrow 4 \rightarrow 5$ paths that will cause block $5$ to branch to the new poison block on the $5 \rightarrow 7$ edge only if block $5$ was reached from a path containing block $3$.

\textit{Case 3:} Store $c$ is also poisoned on the $3 \rightarrow 6$ edge, but here it is safe to prepend $poison(c)$ to the start of block $6$.

\subsection{Merging Poison Blocks} \label{sec:cfg_simplification}
\begin{figure}[t]
\centering
     \centering
     \includegraphics[width=0.47\textwidth]{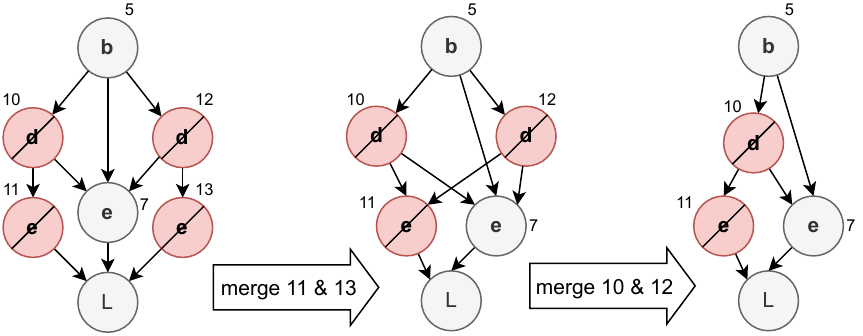}

\caption{Basic blocks with the same list of poison stores and the same immediate successor can be merged in the CU.}
\label{fig:MergePoisonBB}
\end{figure}

Case 1 and 2 of Algorithm~\ref{alg:poison_edge} might create multiple poison blocks for the same store on different CFG edges.
It is possible to merge two poison blocks into one if they contain the same list of poison stores and if they have the same list of immediate successors.
When merging, we keep instructions from just one block.
We apply such merging iteratively after Algorithms \ref{alg:poison} and \ref{alg:poison_edge}.
For example, Figure~\ref{fig:MergePoisonBB} contains a CFG sub-region of our running example from Figure~\ref{fig:HoistingAndPoisoning}.
Algorithm~\ref{alg:poison_edge} inserted poison blocks $10, 11, 12, 13$ to poison stores $d$ and $e$.
Block pairs $(11, 13)$ and $(10, 12)$ can be merged.

\subsection{Speculative Load Consumption}

Speculative loads are relatively easy to support.
To match the order of \texttt{load\_consume} calls in the CU with the order of speculative \texttt{send\_load\_addr} calls in the AGU we can hoist the \texttt{load\_consume} calls to the same block where the corresponding \texttt{send\_load\_addr} were hoisted in the AGU.
Then, the CU can either use the load value or discard it.
After hoisting, we need to update all $\phi$ instructions that use the load value, since the basic block containing the loaded value will have changed.
Alternatively, we can transform $\phi$ instructions using the load value into \texttt{select} instructions.



\section{Safety and Liveness} \label{sec:proof}
We prove that our transformations preserve the sequential consistency of the original program and that they do not introduce deadlock.
Deadlock freedom is a corollary of sequential consistency, so we focus only on the latter.
We show that on every CFG path the order of speculative store requests in the AGU matches the order of store values in the CU, and that the non-poisoned store value sequence in the CU matches the store sequence of the original code.

In the following discussion, we assume blocks with a single store; the proof trivially extends to blocks with multiple stores since all speculative stores in the same block are treated the same.
We also assume that all stores are speculative, since the relative order between non-speculative and speculative stores is guaranteed by definition:
given a a non-speculative store $s_1$ and a speculative store $s_2$, Algorithm~\ref{alg:speculation} will not change the relative program order of $s_1$ and $s_2$, i.e., if $s_1 \prec s_2$ in the original program order, then it is not possible to hoist $s_2$ such that $s_2 \prec s_1$.
This follows from the control dependency definition (\S\ref{sec:LoDAnalysis})---$s_2$ hoisting stops at its LoD control dependency source $srcBB$, which must come after the block containing $s_1$ in topological order.
If $srcBB$ would come after $s_1$ in topological order, then the block containing $s_1$ would also have a LoD control dependency on $srcBB$ and would have been hoisted, which is a contradiction since we assumed that $s_1$ was non-speculative.
A similar argument can be made if $s_2 \prec s_1$ in the original program.

\begin{lemma}[Sequential Consistency] \label{lemma}
Given an ordered list of $n$ speculative store requests $L_a = \{a_0,\, a_1,\, ...,\, a_{n-1}\}$ made in the AGU loop CFG on some fixed iteration $k$, Algorithms \ref{alg:poison} and \ref{alg:poison_edge} transform the CU CFG such that every possible path through its loop CFG on iteration $k$ produces an ordered list of $n$ tagged store values $L_v = \{(v_0,\, p_0),\, (v_1,\, p_1),\, ...,\,$ $(v_{n-1}, p_{n-1})\}$, such that each $(a_i, v_i, p_i) , 0 \leq i < n$ triple corresponds to a $A[a_i] \leftarrow v_i$ store in the original program CFG, and $p_i=1$ (poison bit) if that store is not executed on the path through the original loop CFG on iteration $k$.
\end{lemma}


\begin{proof}

We use a proof by induction on the transformed CFG.

\textit{Base case:} $L_a = \varnothing$ (no speculated requests in the AGU). Algorithm~\ref{alg:poison} does not change the CU CFG. Thus, the order of store addresses in the AGU and store values in the CU trivially matches, $L_a = L_v = \varnothing$.

\textit{Inductive hypothesis:} assume Lemma \ref{lemma} holds at basic block $B_i$ in the current CFG path. 
All store requests $a_j \in L_a$ contained in blocks reached before $B_i$ in the path were matched with the correct store value call $(v_j,\, p_j) \in L_v$, such that $p_j = 1$ if $A[a_j] \leftarrow v_j$ was not executed on the path in the original loop CFG.

\textit{Inductive step:} The next store address in the AGU $L_a$ sequence is $a_{j+1} \in L_a$.
The next store value in the CU CFG path should be $(v_{j+1},\, p_{j+1}) \in L_v$, where $p_{j+1} = 1$ iff the store $A[a_{j+1}] \leftarrow v_{j+1}$ is not reached on the current CFG path in the original program.
Algorithm~\ref{alg:poison} considers the $edge_{src} \rightarrow edge_{dst}$ next.
There are three cases: 
\begin{enumerate}
    \item $edge_{dst} = trueBB$, where $trueBB$ is the block containing the store $A[a_{j+1}] \leftarrow v_{j+1}$ in the original program CFG. In this case, Algorithm~\ref{alg:poison} will not poison this store on this path through the CU CFG, i.e., the next item in the $L_v$ sequence will be the correct $(v_j,\, 0)$. 
    
    \item $edge_{dst} \neq trueBB$ and $trueBB$ is not reachable from $edge_{dst}$, in which case Algorithm~\ref{alg:poison} will insert a poison store on this edge. Algorithm~\ref{alg:poison_edge} will map this poison store to a basic block, with the effect that taking the $edge$ will result in the poison call being executed and control transferring to $edge_{dst}$. The next item in the $L_v$ sequence will be the correct $(v_j,\, 1)$. 
    
    \item $edge_{dst} \neq trueBB$ and $trueBB$ is reachable from $B_l$, in which case Algorithm~\ref{alg:poison} will traverse the path until Case 1 or 2 is matched.
\end{enumerate}

Since Lemma \ref{lemma} holds for the base case, for basic blocks on the path up to $B_i$, and for some successor block of $B_i$, it must hold at any block on the path.
If it holds at any block on the path, it holds for the whole path.
Since a given store request $r$ is poisoned at most once on a given CFG edge and since, by definition of Algorithm~\ref{alg:poison}, any given path will contain at most one edge where $r$ is poisoned, we conclude that Lemma \ref{lemma} holds for all paths.
\end{proof}



\section{Applications} \label{sec:applications}

In this section, we highlight three applications for our work: DAE-based prefetchers in CPUs/GPUs, CGRAs, and specialized accelerators generated from HLS.
In the next section, we choose HLS as an evaluation vehicle due to its simplicity compared to CPU/GPU prefetchers where the evaluation results can easily be polluted by other architectural factors like cache behavior, branch prediction, etc.
However, we emphasize that our speculation support in DAE does not rely on any HLS-specific features and can be applied wherever speculation is combined with the DAE technique.

\subsection{CPU/GPU Prefetchers}

Most existing works on CPU/GPU prefetchers follow the DAE principle and rely on the compiler to decoupled address generation from compute \cite{outrider_decoupled_strands, phloem, desc_cpu, gpu_decoupled, nvidia_hopper_tma, wasp_gpu_prefetch, qin2023roma}.
All of these works suffer from the control-dependency loss of decoupling (LoD) problem (\S \ref{sec:LoDAnalysis}).
The work in \cite{desc_cpu} discusses adding speculation and predicated stores to the CPU microarchitecture to mitigate LoD, but their compiler only supports simple diamond and triangle control flow shapes.
In this paper, we have demonstrated generalized compiler support for speculation in DAE, making these works viable for general control flow and thus applicable to a broader set of codes.

\subsubsection{Example}
The CPU prefetcher proposed in \cite{desc_cpu} (on which most of the other work is based) separates address generation from compute and extends the ISA with \texttt{store\_addr}, \texttt{load\_produce}, \texttt{store\_val}, \texttt{load\_consume}, and \texttt{store\_inv} instructions that can be directly targeted by our compiler.

\subsection{Coarse Grain Reconfigurable Architectures}

A CGRA consists of an array of PEs, each with small memories, connected by a network.
A CGRA compiler is typically co-designed with the hardware, as the PEs are typically statically scheduled.
The job of the compiler is to map the Control/Data Flow Graph (CDFG) to the PEs, and many works follow the DAE technique to tackle the memory wall problem \cite{fifer, Plasticine_rdu, SambaNova, fan2023_europar_decoupled_dataflow, hong2020decoupling_cgra, pellauer2019buffets, dae_cgra_cascade, cgra_dae_softbrain}.
Our work can help mitigate LoD events when mapping to CGRAs.

\subsubsection{Example}
The CGRA proposed in \cite{cgra_dae_softbrain} is an example of a modern streaming dataflow CGRA. 
All communication in the CGRA is FIFO-based, and address generation is explicitly decoupled at compile time into AGUs.
The compiler generates commands to produce address streams, and to consume or produce values.
Control flow is handled with predication and there is a \textit{SD\_Clean\_Port} command to throw away a value from an output port that can be used to implement predicated stores.

\subsection{High-Level Synthesis}

In HLS, the CDFG of an algorithm is implemented directly in hardware following a spatial execution model with the freedom to customize the memory system.
This makes decoupling easier in HLS compared to the temporal CPU/GPU execution model.
HLS-generated accelerators can directly benefit from our work today without any changes, and it is in this domain that we evaluate our implementation in the next section.

Although existing HLS compilers are successful in building non-trivial accelerators for regular code (e.g., \cite{google_vcu_asplos}), their static scheduling techniques are sub-optimal for irregular codes (for the same reason why traditional VLIW compilers were sub-optimal for irregular codes).
Many research works in academia and industry have exploited DAE in HLS to improve the efficiency of HLS-generated accelerators for irregular codes \cite{decoupled_memory_prefetching, dae_accel, decoupled_memory_wawrzynek, brainwave_ms_chung2018serving, ThunderGP, HLS_runahead, szafarczyk_fpt_23, szafarczyk_fpl23, szafa_fpga25}.
By adding compiler speculation support, DAE in HLS can be used on a broader set of codes, which we demonstrate in the next section.

\section{Evaluation} \label{sec:evaluation}

In this section, we answer the following questions: 
\begin{itemize}
\item What is the performance benefit of using a DAE architecture (enabled by our speculation approach) to accelerate codes with LoD control dependencies? 
\item What is the cost of mis-speculation in our approach? 
\item What is the impact on code size (accelerator area usage) of our speculation approach? 
\item What is the scalability for nested control flow, which increases the number of poison stores and blocks?
\end{itemize}

We make our work and evaluation publicly available \cite{spec_dae_zenodo_repo}.

\subsection{Methodology}

We generate algorithm-specific accelerators using HLS targeting an Intel Arria 10 FPGA.
The C codes are taken directly from benchmark suites without adding any HLS-specific annotations (excluding dynamic structures, like queues, that were replaced with HLS-specific libraries).

We use the LLVM-based Intel SYCL HLS compiler \cite{sycl_github} and apply our standard DAE transformation (\S \ref{sec:compiler_Preliminaries}) and our proposed speculation transformation (\S \ref{sec:compiler_support}) as LLVM passes.
The codes use deterministic dual-ported on-chip SRAM capable of 1 read and 1 write per cycle.
To enable out-of-order loads, we use a load-store queue (LSQ) designed for HLS (load/store queue sizes of 4/32), which is commonly found on accelerators for irregular codes \cite{desc_cpu, dac_runtime_dep_check_with_shiftreg, szafarczyk_fpl23, josipovic_dynamatic_2022}.

We report cycle counts from ModelSim simulations.
We do not report circuit frequency since our approach does not affect the critical path (see \cite{spec_dae_zenodo_repo} for such details).
Area usage is obtained after place and route using Quartus 19.2.

\subsubsection{Baselines}
For each benchmark, we synthesize the following architectures which represent current state-of-the-art approaches to HLS :
\begin{itemize}
    \item \texttt{STA}: the default, industry-grade approach using static scheduling  \cite{sycl_github}. Loads that cannot be disambiguated at compile time execute in order.
    \item \texttt{DAE}: a DAE architecture without speculation. OoO loads are enabled by an LSQ. This is the state-of-the-art approach to irregular codes in academia \cite{szafarczyk_fpl23}, but it suffers from control-dependency LoDs.
    \item \texttt{SPEC}: the same as \texttt{DAE}, but with our speculation technique which mitigates control-dependency LoDs.
    \item \texttt{ORACLE}: the same as \texttt{DAE}, but all LoD control dependencies are removed manually from the input code. The \texttt{ORACLE} results are wrong, but give a bound on the performance of \texttt{SPEC} and show its area overhead. 
\end{itemize}

\begin{figure}[t]
\centering
\includegraphics[width=0.47\textwidth]{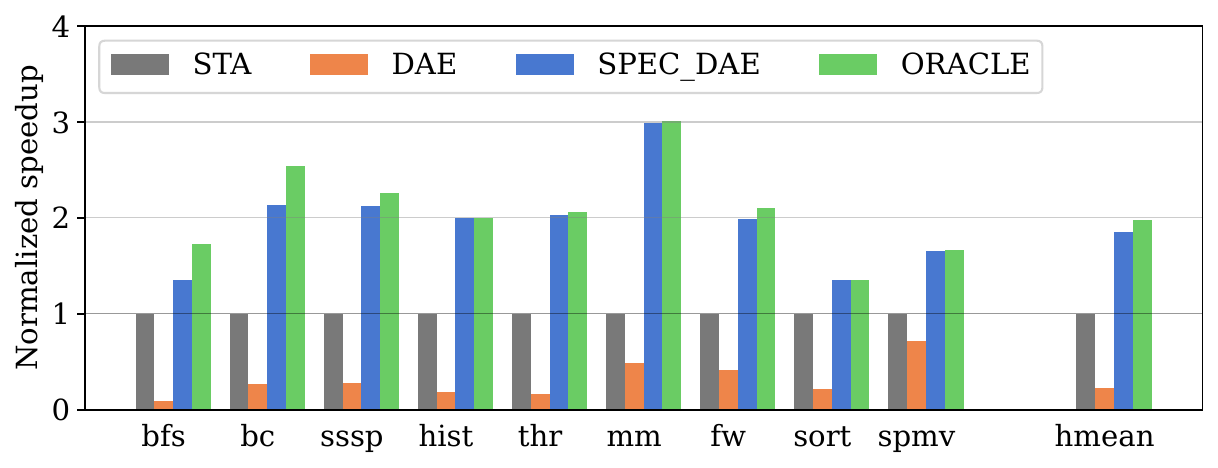}

\caption{Performance of \texttt{DAE}, \texttt{SPEC} and \texttt{ORACLE} normalized to \texttt{STA}. \texttt{SPEC} achieves an average $1.9\times$ (up to $3\times$) speedup.}
\label{fig:PlotPerformance}
\end{figure}

\begin{table*}[!t]
\centering

\renewcommand{\arraystretch}{1}

\caption{Absolute performance and area usage of \texttt{STA} \cite{sycl_github}, \texttt{DAE} \cite{szafarczyk_fpl23}, \texttt{SPEC}, and \texttt{ORACLE} accelerators. (*bc uses two LSQs).}

\label{table:benchmarkTable}
\centering
\begin{tabular}{l | cc | c | cccc | cccc }
\hline
\multirow{2}{*}{\textbf{Kernel}} & \multicolumn{2}{c|}{\textbf{Poison}} & \textbf{Mis-spec.} & \multicolumn{4}{c|}{\textbf{Cycles}} & \multicolumn{4}{c}{\textbf{Area (ALMs \cite{alm})}} \\


& Blocks & Calls & \textbf{Rate} &
STA & DAE & SPEC & ORACLE &
STA & DAE & SPEC & ORACLE \\
\hline
\hline

bfs & 1 & 1 & 95\% & 37,243 & 398,616 & 27,561 & 21,569 & 7,361 & 7,525 & 13,404 & 13,706 \\
bc & 2 & 2 & 95\%, 82\% * & 109,061 & 406,178 & 51,109 & 42,942 & 9,709 & 10,859 & 16,582 & 16,558 \\
sssp & 1 & 1 & 95\% & 108,995 & 391,426 & 51,227 & 48,208 & 10,565 & 11,668 & 17,426 & 17,395 \\
hist & 1 & 1 & 2\% & 2,061 & 11,100 & 1,033 & 1,031 & 2,391 & 2,807 & 3,117 & 3,137 \\
thr & 1 & 3 & 97\% & 2,131 & 13,147 & 1,052 & 1,034 & 5,662 & 6,144 & 6,278 & 6,622 \\
mm & 1 & 2 & 31\% & 12,164 & 25,125 & 4,069 & 4,044 & 5,076 & 4,986 & 7,813 & 7,528 \\
fw & 1 & 1 & 85\% & 6,821 & 16,485 & 3,433 & 3,238 & 3,407 & 4,210 & 4,008 & 4,007 \\
sort & 1 & 2 & 49\% & 2,358 & 11,109 & 1,748 & 1,746 & 2,814 & 4,361 & 5,260 & 5,269 \\
spmv & 1 & 1 & 32\% & 13,319 & 18,693 & 8,028 & 7,984 & 3,895 & 5,085 & 4,416 & 4,336 \\

\hline
\hline

\multicolumn{4}{r|}{Harmonic Mean:} & 1 & 3.2 & 0.51 & 0.48 & 1 & 1.16 & 1.42 & 1.36 \\


\end{tabular}

\end{table*}

\subsubsection{Benchmarks}

DAE architectures optimize the latency between memory and compute and are most beneficial for memory-bound codes \cite{desc_cpu}, especially codes with an irregular memory access pattern that prevents static prefetching \cite{outrider_decoupled_strands}.
We evaluate nine such benchmarks from the graph and data analytics domain, using the GAP graph benchmark suite \cite{gap_benchmark_suite} and an HLS benchmark suite \cite{hls_dyn_bench_jianyi_cheng} of irregular programs.
We select only codes that can benefit from our \texttt{SPEC} approach, i.e., codes with LoD control dependencies:
\begin{itemize}
    \item \texttt{bfs}: breadth-first traversal through a graph.
    \item \texttt{bc}: betweenness centrality of a single node in a graph.
    \item \texttt{sssp}: single shortest path from a single node to all other nodes in a graph using Dijkstra's algorithm.
    \item \texttt{hist}: histogram, similar to Figure~\ref{fig:DecoupledAccessExecute_b} (size 1000).
    \item \texttt{thr}: zeroes RGB pixels above threshold (size 1000).
    \item \texttt{mm}: maximal matching in a bipartite graph (2000 edges).
    \item \texttt{fw}: Floyd-Warshall distance calculation of all node-to-node pairs in a dense graph (10$\times$10 distance matrix).
    \item \texttt{sort}: using bitonic mergesort (size 64).
    \item \texttt{spvm}: sparse vector matrix multiply (20$\times$20 matrix).
\end{itemize}
For the graph codes (\texttt{bfs}, \texttt{bc}, \texttt{sssp}) we use a real-world graph \texttt{email-Eu-core} with 1005 nodes and 25,571 edges.

\subsection{Performance}

Figure~\ref{fig:PlotPerformance} reports normalized speedups of each technique over \texttt{STA}. 
Our \texttt{SPEC} approach gives on average a $1.9\times$ (and up to $3\times$) speedup over \texttt{STA}.
This is within 5\% of the \texttt{ORACLE} performance.
In contrast, \texttt{DAE} without speculation sees a dramatic performance degradation over \texttt{STA}, because the AGU, DU, CU communication is sequentialized.

\subsubsection{Mis-speculation Cost}
The \texttt{SPEC} and \texttt{ORACLE} performance gap is highest on the \texttt{bfs} and \texttt{bc} codes, because of its deep pipeline between the load and store that form a RAW hazard.
The deep pipeline means that more store allocations need to be held by the LSQ to guarantee perfect pipelining \cite{lsq_sizing_dynamatic}.
This, together with a high mis-speculation rate in these benchmarks (Table \ref{table:benchmarkTable}), can cause the LSQ to fill with store addresses that are mis-speculated, potentially stalling later loads that have to wait for future store addresses to arrive.
This problem can be solved by increasing the store queue size in the LSQ.
The increased number of requests and the need for more buffering is one of the limitations of our approach. 
Codes with a shallower pipeline that do not need large LSQ sizes have no mis-speculation penalty.

To prove this, we choose three benchmarks where we can instrument the input data so that we can vary the mis-speculation rate.
Table \ref{table:misspecTable} shows how the mis-speculation cost changes as the mis-speculation rate increases.
As can be seen, there is no correlation between the mis-speculation rate and cost, with the slight variability in clock cycle counts attributable to the subtle difference in the number of true RAW hazards due to the varying data distribution.
\begin{table}[h!]
\centering

\renewcommand{\arraystretch}{1}

\caption{\texttt{SPEC} cycle counts as mis-speculation rate changes.}

\label{table:misspecTable}
\centering
\begin{tabular}{l | cccccc | c  }
\hline
\multirow{2}{*}{\textbf{Kernel}} & \multicolumn{6}{c|}{\textbf{Mis-speculation rate}} & \multirow{2}{*}{$\sigma$} \\


& 0\% & 20\% & 40\% & 60\% & 80\% & 100\% & \\
\hline

hist & 1044 & 1013 & 1029 & 1029 & 1012 & 1051 & 16 \\
thr & 1082 & 1109 & 1047 & 1073 & 1058 & 1071 & 21 \\
mm & 4107 & 4096 & 4074 & 4063 & 4106 & 4081 & 18 \\

\end{tabular}

\end{table}

\subsection{Code Size} \label{sec:Code_Size_Increase}

Our speculation approach can increase the number of blocks in the CU, especially for codes with deeply nested control flow. 
In HLS, an increased number of blocks can result in a higher area usage due to larger scheduler complexity \cite{modulo_sched}.

Table \ref{table:benchmarkTable} shows the absolute area usage of all accelerators.
We observe virtually no area overhead of \texttt{SPEC} over \texttt{ORACLE} on the evaluated benchmarks.
This is because most of the codes have at most two control-flow nesting levels where new poison blocks are inserted, and sometimes it is possible to reduce the number of blocks using our merging technique (e.g., two poison blocks in \texttt{mm} merged into one).


\subsubsection{Impact of Nested Control Flow on Area Usage}

To give a more meaningful measure of how nested control flow impacts the area overhead of our \texttt{SPEC} approach, we create a synthetic benchmark template where we can tune the number of poison blocks generated by \texttt{SPEC}:

\begin{algorithmic}
\If {$x > 0$}
    \State $store_1$
    \If {$x > 1$}
        \State $store_2$
        \State \textbf{if} $x > 2$ \textbf{then}  ...
    \EndIf
\EndIf
\end{algorithmic}
Each nesting level in this template will result in one poison block in the \texttt{SPEC} architecture.
With $n$ stores, and assuming one store per nesting level, there will be $n$ poison blocks and $\sum_{i=1}^{n} i = \frac{n \times (n+1)}{2}$ poison calls.

Figure~\ref{fig:PlotNumBasicBlocks} shows how the area and performance overhead of \texttt{SPEC} over \texttt{ORACLE} changes as more poison blocks are needed.
The performance overhead is close to 0\% and does not change with more poison blocks.
The area overhead of the AGU unit is similarly close to 0\%, because \texttt{SPEC} hoists stores out of the \textit{if}-conditions, causing the blocks to be deleted. 
The area overhead of the CU unit grows by a few percent ($<5\%$) with each added poison block, but even for the pathological case of eight nested \textit{if}-conditions the overhead is below $25\%$.
In real codes, with more compute and lower control-flow nesting, the area overhead of \texttt{SPEC} should be minimal.

\begin{figure}[t]
\centering
\includegraphics[width=0.475\textwidth]{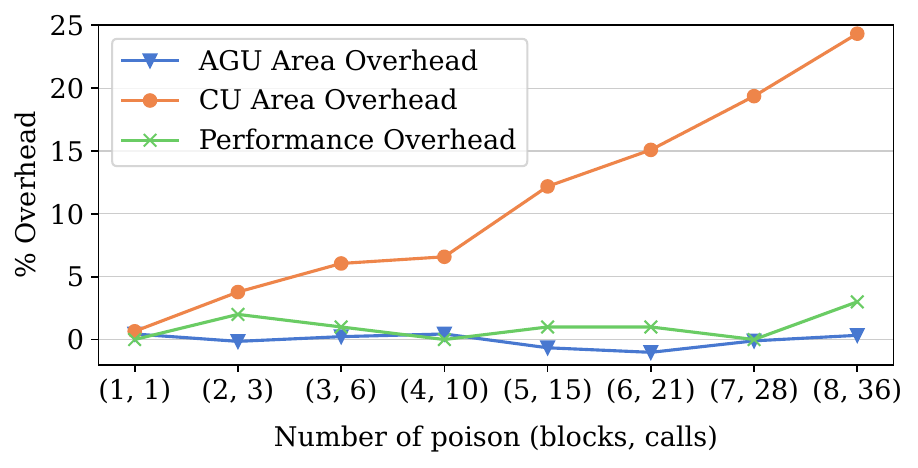}

\caption{Change in area and performance overhead of \texttt{SPEC} over \texttt{ORACLE} as the number of poison blocks and calls grows.}
\label{fig:PlotNumBasicBlocks}
\end{figure}




\section{Related Work}

\textit{Program slicing} is used beyond DAE architectures.
Decoupled Software Pipelining (DSWP) \cite{dswp} is a popular transformation that decouples strongly connected components in the program dependence graph into separate pipeline stages mapped over multiple PEs communicating via FIFOs.
The PEs can be CPU threads, or pipeline stages in an accelerator generated by HLS \cite{cgpa_dswcp_for_hls}.
Control dependent pipeline stages in DSWP can also be executed speculatively, although stages with memory operations require versioned memory \cite{spec_dswp}.

\textit{Control speculation} has its roots in compilers for VLIW machines.
Instruction scheduling in HLS is very similar to VLIW scheduling (no hardware support for speculation, static mapping to functional units, etc.), with many algorithms like modulo-scheduling and \textit{if}-conversion originally developed for VLIW directly applicable to HLS \cite{modulo_sched, if_conversion, park1991predicated}.
Most recently, predicated execution in the form of gated SSA was proposed for HLS with speculation support \cite{unified_memory_dependency_spec_hls}.
The speculation support in this and other works requires costly recovery on mis-speculation \cite{Josipovic_Guerrieri_Ienne_2019, Thielmann_Load_Speculation, ctrl_spec_and_if_conversion, vliw_sentinel_scheduling, hls_cancel_tokens_koch, dsagen}.
Efficiently squashing speculative computation on the wrong paths in a spatial dataflow architecture is hard, because the architectural state is distributed  \cite{Budiu_Artigas_Goldstein_2005}.
Our speculative DAE sidesteps this issue, not requiring any recovery: we speculate early (run ahead) in the AGU, and later handle mis-speculations in the CU by taking an appropriate path in its CFG.


\textit{Control-flow handling in GPUs} is usually implemented via \textit{predication}.
The algorithms used to calculate predicate masks and re-convergence points resemble our work \cite{chap_gpu}.
The SIMT stack approach in GPUs pushes predicate masks onto a stack when entering a control-flow nesting level, and pops when exiting. 
Our Algorithm~\ref{alg:speculation} implementing speculative requests can be seen as a pass through the CFG with only push operations, where the push is onto individual stacks of control-dependency sources.
Dually, our inserting of poison calls in Algorithm~\ref{alg:speculation} can be seen as a pass through the CFG with only pop operations where the placement of the pops follows a certain policy just like modern SIMT compilers follow different policies to prevent SIMT deadlock and livelock, or to improve performance \cite{MIMD_synchronization_on_SIMT_architectures}, instead of popping at the immediate post-dominator.


\section{Conclusion}

We have presented general compiler support for speculative memory operations in DAE architectures that tackles the LoD problem resulting from control dependencies.
We have proposed CFG transformations implementing speculation in the address generation slice, and poisoning of mis-speculations in the compute slice, with a proof of correctness.

We have presented three applications where our work improves support for the efficient execution of irregular codes: DAE-based CPU/GPU prefetchers that require compiler support, CGRA architectures, and HLS-generated specialized accelerators. 
We have evaluated our work on HLS-generated accelerators, showing an average $1.9\times$ (up to $3\times$) speedup over non-DAE accelerators on a set of irregular benchmarks where DAE is not possible without our speculation.
Our approach has no mis-speculation cost and a small code size footprint, scaling well to deeply nested control flow.

Future work could investigate vector-parallelism support by filling a vector of speculative requests in the AGU and producing a store mask in the CU, similar to the recent work on decoupled vector runahead prefetching in CPUs \cite{decoupled_vector_runahead}.


\section*{Data-Availability Statement}

We make our work and evaluation publicly available \cite{spec_dae_zenodo_repo}.

\begin{acks}
We thank Intel for access to FPGAs through their DevCloud.
This work was supported by a UK EPSRC PhD scholarship. 
\end{acks}

\bibliographystyle{ACM-Reference-Format}
\bibliography{references}


\end{document}